\documentclass[preprint,12pt]{elsarticle}

\usepackage[T1]{fontenc}

\usepackage{amssymb}
\usepackage{amsthm}
\usepackage{amsmath}
\usepackage{bm}
\usepackage[linguistics]{forest}
\usepackage{algorithm}
\usepackage{algpseudocode}
\usepackage{multirow}
\usepackage{subcaption}

\usepackage{xcolor}
\usepackage{tabularray}
\usepackage{ulem}

\newcommand{\M}{\kappa(l_1, \dotsc, l_k)}

\newcommand{\N}{\mathbb{N}}

\newcommand{\R}{\mathbb{R}}
\newcommand{\F}{\mathbb{F}}
\newcommand{\bigO}{\mathcal{O}}

\newcommand{\str}[1]{\operatorname{str}(#1)}
\newcommand{\nut}{N \cup T}  

\newcommand{\opcija}[2]{#1\;\; [#2]}

\theoremstyle{plain}
\newtheorem{thm}{Theorem}[section]
\newtheorem{lema}[thm]{Lemma}
\newtheorem*{cor}{Corollary}

\theoremstyle{definition}
\newtheorem{definition}{Definition}[section]
\newtheorem{ex}{Example}
\newtheorem{note}{Note}

\begin{document}

\begin{frontmatter}



\title{$P$(Expression|Grammar): \\ Probability of deriving an algebraic expression \\ with a probabilistic context-free grammar}


\author[fmf]{Urh Primožič}
\author[fmf,jsi]{Ljupčo Todorovski}
\author[fmf,jsi]{Matej Petković\corref{x}}

\affiliation[fmf]{organization={University of Ljubljana, Faculty of Mathematics and Physics},
            addressline={Jadranska 21}, 
            city={Ljubljana},
            postcode={SI-1000}, 
            country={Slovenia}}

\affiliation[jsi]{organization={Jožef Stefan Institute, Department of Knowledge Technologies},
            addressline={Jamova 39}, 
            city={Ljubljana},
            postcode={SI-1000}, 
            country={Slovenia}}
\cortext[x]{Corresponding author, \texttt{matej.petkovic@fmf.uni-lj.si}}

\begin{abstract}
Probabilistic context-free grammars have a long-term record of use as generative models in machine learning and symbolic regression. When used for symbolic regression, they generate algebraic expressions. We define the latter as equivalence classes of strings derived by grammar and address the problem of calculating the probability of deriving a given expression with a given grammar. We show that the problem is undecidable in general. We then present specific grammars for generating linear, polynomial, and rational expressions, where algorithms for calculating the probability of a given expression exist. For those grammars, we design algorithms for calculating the exact probability and efficient approximation with arbitrary precision.
\end{abstract}




\end{frontmatter}



\section{Introduction}

Every language (natural or programming) is defined by grammar, i.e., a set of allowed symbols and rules that specify how the symbols can form longer strings (e.g., words or sentences). In context-free grammars (CFGs)~\cite{sipser2006}, the rules must not take the context of symbols into account, so rules such as \textit{i before e except after c} are not possible (since \textit{after c} defines the context of the rule). Thus, CFGs are inappropriate for natural languages and are regularly used to specify programming languages. Despite that, CFGs are an established and widely used tool for natural language processing, e.g.,~\cite{nltk, nlp-cfg-seki, nlp-cfg-thorsteinsson}.

Their extension, probabilistic (also known as stochastic) context-free grammars (PCFGs), assign probabilities to the rules, leading to the probability of deriving a string. \textit{Consistent} PCFG defines a probability distribution over the space of strings in the grammar language, i.e., the probabilities of the strings derived by the consistent grammar sum up to 1. The ability of CFGs to specify languages and PCGs to specify probability distributions over strings in the language make them a common choice for generative models in machine learning~\cite{kusner2017,duvwnaud2013}, and also symbolic regression~\cite{todorovski1997,brence2021}. 

Given a set of observations $(\bm{x}, y)$, the goal of symbolic regression is to discover a closed-form equation of the form $y = f(\bm{x})$ that sufficiently well explains the target variable $y$ in terms of independent variables $\bm{x} = (x_1, \dots, x_n)$. PCFGs can be used to specify the probability distribution over the space of candidate expressions for the right-hand side of the equation. Then, equations can be discovered by following the generate-and-test paradigm~\cite{brence2021}. In the generating step, expressions with generic constants are sampled, e.g., $c_1 x_1 + c_2 x_2$ with the constants $c_1$ and $c_2$. In the second (i.e., testing) step, numerical optimization is used to fit the values of $c_1$ and $c_2$ against the given set of observations. The probabilities of the grammar rules are set so that generating simpler arithmetic expressions is more probable than generating complex ones.

This use case of PCFGs is the main focus of the paper. We are interested in PCFGs generating algebraic expressions. More specifically, we would like to compute the probability that a given PCFG derives a given expression $e$. In general, such a computation can involve summing up the probabilities of an infinite number of derivations of $e$, which happens in the following two scenarios. On the one hand, the grammar might be \textit{ambiguous}, i.e., it can derive $e$ in more than one way.
On the other hand, multiple stings can be equivalent to the same algebraic expression. For example, strings \texttt{c x + d} and \texttt{d + c x} are different but equivalent to the same expression $c x + d$. The probability of the latter is the sum of the probabilities of all the strings derived from the grammar being equivalent to $c x + d$.

Thus, the article's main aim is to design an algorithm for calculating the probability of deriving an expression $e$ from an expression-generating PCFG $G$. The purpose is achieved through three key contributions. The first contribution of the article is a formal definition of the equivalence relation $\sim$ between strings that relates strings that are equivalent to the same expression. The second contribution is the proof that an algorithm that would compute the probability of an equivalence class $[w]_\sim$ of a given string $w$ for any given PCFG does not exist. The article's third contribution is an algorithm for specific PCFGs generating linear, polynomial, and rational expressions. We analyze the algorithm's computational complexity and propose an approximate version that efficiently approximates the actual probability for cases when the exact calculation is infeasible.

The rest of the article is organized as follows. Section~2 introduces CFGs, PCFGs, and the notation used in the article. In Section~3, we introduce expression-generating grammars and their use for symbolic regression, formally define an equivalence relation among strings of an expression-generating grammar, and define expressions as equivalence classes of the relation. Section~4 proves that the algorithm for calculating the probability of deriving a given expression with any given PCFG does not exist and presents exact and approximate algorithms for grammars generating linear, polynomial, and rational expressions. Section~5 concludes the paper with a summary and outline of directions for further research.

\section{Probabilistic context-free grammars}

In this section, we first define context-free grammars and language of strings that can be derived with a given grammar. In continuation, we define probabilistic context-free grammars and probabilities of strings in the grammar language. Finally, we discuss the issue of calculating a probability of a string in the grammar language. Table~\ref{tab:notation} enlists and defines the semantics of the notation symbols that we use in the rest of the article.


\begin{table}[!htb]
    \centering
    \begin{tabular}{l|p{0.75\textwidth}}
        Symbol & Definition and semantics \\
        \hline
        $V^*$ & the set of all strings (finite sequences) with elements from $V$, i.e., $V^*=\{a_1 \dotsm a_n \mid a_i \in V \land n \in \N_0 \}$ \\
        $\varepsilon$ & the empty string, $\varepsilon \in V^*$ for any set $V$ \\
		\hline
        $T$ & the set of terminal symbols, also terminals \\
        \texttt{a}, \texttt{x}, \texttt{y} & arbitrary terminals, elements of $T$ \\
        $w$, $u$, $v$ & arbitrary strings of terminals (words), elements of $T^*$ \\
		\hline
        $N$ & the set of non-terminal symbols, also non-terminals \\
        $S$ & the starting non-terminal, element of $N$ \\
        $A$, $B$, $C$ & arbitrary non-terminals, elements of $N$ \\
		\hline
		$X$, $Y$, $Z$ & arbitrary non-terminals or terminals, elements of $\nut$ \\
		$\alpha$, $\beta$, $\gamma$ & arbitrary strings of terminals and non-terminals, elements of $(\nut)^{*}$ \\
		\hline
        $R$ & the set of grammar production rules \\
		$\rho$, $A \to \alpha$ & a grammar production rule, also rule, element of $R$ \\
		\hline
		$G$ & deterministic or probabilistic context-free grammar \\
		$\tau$ & grammar parse tree \\
        $p$, $q$, $r$ & probabilities \\
		\hline
    \end{tabular}
    \caption{The definitions and semantics of the notation symbols used throughout the article.}
    \label{tab:notation}
\end{table}

Let us first introduce the formal definition of a context-free grammar.
\begin{definition}
\label{def:cfg}
Let $N$ and $T$ be nonempty, finite disjoint sets and $S$ a distinguished element of $N$. Let $R \subseteq N \times (\nut{})^*$ be a left-total relation. A tuple $G = (N, T, S, R)$ is called \emph{context-free grammar (CFG)} with a set of non-terminal symbols $N$, terminal symbols $T$, start symbol $S$ and a set of production rules $R$.
\end{definition}
\noindent Note that in literature on context-free grammars, e.g., the standard textbook~\cite{sipser2006}, alternative symbols and notions can be used. The elements of $N$ can also be referred as variables (and denoted by $V$), the set $T$ can also be an alphabet (denoted by $\Sigma$), and the strings, derived by the grammar, are referred to as expressions. However, since the focus of the paper are grammars that generate mathematical expressions consisting of variables and constants, we are using the notation introduced in Table~\ref{tab:notation} and Definition~\ref{def:cfg}.

A production rule $(A, \alpha) \in R$ is often written as $A \to \alpha$, conveying the notion that $A$ can be rewritten with $\alpha$. For a rule $A \to \alpha \in R$ and strings $\beta, \gamma \in (\nut)^*$, we also write $\beta A \gamma \to \beta \alpha \gamma$ meaning that we can \emph{derive} the string $\beta \alpha \gamma$ from $\beta A \gamma$ using the rule $A \to \alpha$. The set of production rules $\{A \to \alpha_1, \dotsc, A \to \alpha_n\}$ for a non-terminal $A$ can be more compactly written as $A \to \alpha_1 \mid \cdots \mid \alpha_n$.

Given a CFG, we derive (also generate) a string of terminals (also word) as follows. We start with $S$ and choose an arbitrary rule $S \to \alpha \in R$ and replace $S$ with $\alpha$. In the iterative step, we choose a non-terminal symbol $A \in \alpha$, a rule $A \to \beta$, and replace $A$ with $\beta$. We repeat the iteration until the resulting string $\beta$ consists of terminals only. Since the relation $R$ is left-total (recall Definition~\ref{def:cfg}), we can always find at least one rule applicable in the iterative step. Below, we give a more formal definition of string generation.
\begin{definition}
\label{def:language}
A CFG $G = (N, T, S, R)$ \emph{derives} a string $w \in T^*$, if there exists a finite sequence $(\alpha_i)_{i = 1}^n$ of elements of $(\nut)^*$, such that $\alpha_1 = S$, $\alpha_n = w$, and $\alpha_i \to \alpha_{i+1}$, for all $1 \leq i < n$. The set $L(G)$ of all strings that $G$ derives is referred to as the \emph{language} of $G$.
\end{definition}

\noindent Let us illustrate the above definitions on a simple example.

\begin{ex}
\label{ex:basic-grammar}
Let $G=(\{S, A, B, C\}, \{\texttt{x},\texttt{y}\}, S, R)$ be a CFG with the following five production rules:
\begin{align*}
    S & \to AB, \\
    A & \to \mid \texttt{x}, \\
    B & \to C \mid \texttt{y}, \\
    C & \to \texttt{y}.
\end{align*}
The language of this grammar is a singleton $L(G) = \{\texttt{xy}\}$. A possible generating sequence for \texttt{xy} is $S \to AB \to \texttt{x}B \to  \texttt{xy}$.
\end{ex}

\noindent The generating sequence is often depicted with a parse (also derivation) tree.

\begin{definition}
\label{def:tree}
\emph{Parse tree} is a directed, rooted tree, such that 
    \begin{itemize}
    \item its nodes are symbols from $\nut{}$: its root is $S$, its internal nodes are non-terminals from $N$, its leaves are terminals from $T$;
    \item children of an internal node $A \in N$ are the symbols in the string $\alpha \in (\nut)^*$, for some production rule $A \to \alpha\; \in R$, ordered from left to right in the same order as in $\alpha$.
    \end{itemize}
The string derived by a parse tree is the left-to-right sequence of its leaves, and is always an element of $L(G)$. For each $w \in L(G)$, there exist at least one parse tree $\tau$ generating $w$, which is denoted by $\str{\tau} = w$. If there is more than one parse tree for at least one string in $L(G)$, the grammar $G$ is \emph{ambiguous}. The set of all parse trees for a grammar $G$ is denoted by $\Psi(G)$.
\end{definition}

\begin{ex}
\label{ex:basic-tree}
Let $G$ be the grammar from Example~\ref{ex:basic-grammar}. Then, the set $\Psi(G)$ of parse trees derived by grammar $G$ consists of the following two elements: 
$$
\begin{forest}
xx/.style={edge={white,line width=0pt}},
 [\hphantom{S}, xx
 [ $\tau_1 $\text{ =} , xx]
 ]
 \end{forest}
\begin{forest}
 [$S$
 [$A$ [\texttt{x}]]
 [$B$ [\texttt{y}]]
 ]
\end{forest}\quad\quad
\begin{forest}
xx/.style={edge={white,line width=0pt}},
 [\hphantom{S}, xx
 [\text{and}, xx]
 ]
 \end{forest}
\quad\quad
\begin{forest}
xx/.style={edge={white,line width=0pt}},
 [\hphantom{S}, xx
 [ $\tau_2 $\text{ =} , xx]
 ]
 \end{forest}
\begin{forest}
 [$S$
 [$A$ [\texttt{x} , tier=T]]
 [\hphantom{o}$B$   . [$C$ [\texttt{y}, tier=T]]]
 ]
\end{forest}
$$
Since, $\str{\tau_1} = \str{\tau_2} = \texttt{xy}$, $G$ is ambiguous.
\end{ex}

\subsection{Probabilistic context-free grammars}

We can extend context-free grammars to probabilistic context-free grammars by assigning probabilities to the production rules as follows.
\begin{definition}
\label{def:pcfg}
Let $G=(N, T, S, R)$ be a context-free grammar and $P \colon R \to [0,1]$ a mapping, such that
$$\sum_{i = 1}^{n} P(A \to \alpha_i) = 1,$$ 
for all $A\in N$ and the corresponding production rules $A\to \alpha_1 \mid \dots \mid \alpha_n$.
A pair $(G, P)$ is a \emph{probabilistic context-free grammar (PCFG)}.
\end{definition}

\noindent The probability distribution over the rules for $A \in N$ defines the probabilities with which we select one of the rules $A\to \alpha_1 \mid \dots \mid \alpha_n$ for rewriting the symbol $A$. We usually write rules and their probabilities as 
$$A \to \alpha\; [P(A \to \alpha)] \quad \text{or} \quad A \xrightarrow{P(A \to \alpha)} \alpha,$$
and conveniently extend the domain of $P$ to $N\times (\nut{})^*$ with assuming that $P(\rho) = 0$ for $\rho \notin R$. 
If we further assume that the production rules for deriving a string are being chosen independently from each other, we can define the probability of a given parse tree as follows.
\begin{definition}
\emph{Probability of a parse tree} $\tau$ from $G$ is defined as a product of probabilities of all the production rules in $\tau$, i.e.,
    $$
    P(\tau) = \prod_{\rho \in R} P(\rho)^{f_\tau(\rho)},
    $$
where  $f_\tau(\rho)$ denotes the frequency (number of appearances) of rule $\rho$ in $\tau$. 
\end{definition}

\begin{ex}
\label{ex:basic-proba}
Let $G$ be a PCFG with start symbol $S$, terminal symbol \texttt{x} and two production rules $S \to SS\ [p] \mid \texttt{x}\ [1-p]$. String \texttt{x} is derived by the following two parse trees:
$$
\begin{forest}
xx/.style={edge={white,line width=0pt}},
 [\hphantom{S}, xx
 [ $\tau_1 $\text{ =} , xx]
 ]
 \end{forest}
\begin{forest}
 [$S$
[$S$ [$S$ [\texttt{x}, tier=T]]  [$S$ [\texttt{x}, tier=T]]]
[$S$ [\texttt{x}, tier=T]]
 ]
\end{forest}\quad\quad
\begin{forest}
xx/.style={edge={white,line width=0pt}},
 [\hphantom{S}, xx
 [\text{and}, xx]
 ]
 \end{forest}
\quad\quad
\begin{forest}
xx/.style={edge={white,line width=0pt}},
 [\hphantom{S}, xx
 [ $\tau_2 $\text{ =} , xx]
 ]
 \end{forest}
\begin{forest}
 [$S$
[$S$ [\texttt{x}, tier=T]]
[$S$ [$S$ [\texttt{x}, tier=T]]  [$S$ [\texttt{x}, tier=T]]]
 ]
\end{forest}
\begin{forest}
xx/.style={edge={white,line width=0pt}},
 [\hphantom{S}, xx
 [\hphantom{S}, xx
 [\text{.}, xx]
 ]
 ]
 \end{forest}
$$
In both trees, the rule $S \to SS$ appears twice and the rule $S\to \texttt{x}$ is used three times. Thus, $P(\tau_1) = P(\tau_2) = p^2(1-p)^3$.
\end{ex}

\noindent Following the example above, 
we can compute the probability of deriving the string \texttt{x}: $P(\texttt{x}) = P(\tau_1) + P(\tau_2) = 2 p^2 (1-p)^3$. In general, we have
\begin{cor}
\emph{Probability of deriving a string} $w$ with grammar $G$ equals the sum of probabilities of all parsing trees that generate $w$, i.e.,
\begin{equation}
    \label{eq:p-word}
      P(w) = \sum \limits _{\tau:\; \str{\tau} = w} P(\tau).
\end{equation}
\end{cor}
\noindent Two remarks regarding the above formula need to be given. First, the sum $\sum_{w \in L(G)} P(w) = \sum_{\tau \in \Psi(G)} P(\tau)$ is finite, but might be strictly less than $1$, i.e., a string deriving process does not stop with a positive probability. See \ref{sec:appendix:examples} for an example. Second, the formula cannot be immediately implemented as an algorithm for calculating probabilities, since the sum might iterate over an infinite number of parse trees. In the next section, we further discuss the problem of calculating the probability of deriving a string $w \in L(G)$ with $G$.

\subsection{Probability of deriving a string}

We are interested in an algorithm that for a given input, consisting of a grammar $G$ and a string $w$, outputs $P(w)$. The sum~\eqref{eq:p-word} over all the parse trees deriving $w$ can not be immediately used in cases when infinite number of parse trees deriving $w$. Since $\nut{}$ and $R$ are finite sets, the set of parse trees generating $w$ can only be infinite, if at least one production rule leads to a branch with unlimited depth. In other words, there is a sequence of rules that overwrites a non-terminal symbol $A$ with a string $\alpha$ containing $A$. In the case of $G = (\{S\}, \{\texttt{x}\}, S, S \to S \mid \texttt{x})$, this sequence is simply $S\to S$.
\begin{definition}
\label{def:cycle}
    A \emph{cycle} in a parse tree is a branch $A \to \dotsb \to \alpha A\beta$, for some $\alpha, \beta\in (\nut{})^*$.
    If $\alpha = \beta = \epsilon$, the cycle is \emph{linear}. We define the \emph{length $|\cdot|$ of a cycle} as the number of production rules in the cycle.
\end{definition}
\begin{ex}
Both linear and non-linear cycles can lead to an infinite number of parse trees. For example, the rules $S\to S \mid \texttt{x}$ cause a linear cycle of length $|S\to S| = 1$, and generate only the string \texttt{x}. The same string can be generated by the rules  $S\to SS \mid \texttt{x} \mid \epsilon$, where the cycle $S\to SS$ is not linear.
\end{ex}

We can see that applying a non-linear cycle increases the length of the string. Therefore, a non-linear cycle can lead to an infinite number of parse trees all deriving the same string, only if some null rule (rule of the form $A\to\epsilon$) is also present in the grammar. It was shown in~\cite{etessami} that in such a case, a general exact algorithm for computing $P(\epsilon)$ does not exist. Moreover, we show (see~\ref{thm:p-epsilon-polynomial-zeros}) that the computation of $P(\epsilon)$ is at least as hard as finding the roots of polynomials over $\R$.

If we assume a probabilistic context-free grammar without null rules, the standard technique for calculating $ P(w) $ for a given grammar $ G $ is to transform G to its Chomsky-normal form (CNF). We can then apply the well-known dynamic programming parser Cocke-Kasami-Younger (CKY)~\cite{sipser2006}. While the latter was initially proposed for deterministic context-free grammar, its extensions for probabilistic CFGs are also available~\cite{chappelier1998}. The combination of CNF and CKY can calculate the probability of deriving arbitrary given $ w $ with any given grammar $ G $ (without null rules) in polynomial time.

In~\ref{app:removing-linear-cycles}, we propose a series of transformations that strips the linear cycles from a probabilistic context-free grammar, for which Eq.~\eqref{eq:p-word} will include a finite number of parse trees for any given string $w$. We show that such a transformation is applicable to any given PCFG. However, the transformed grammar is not paired with an efficient algorithm for enumerating the parse trees deriving a given string, as is the case for the combination of CNF and CKY.


\section{Expression-generating grammars}\label{sec:expression-generating-grammars}
We say a  $G$ grammar is \emph{expression-generating} if every word from its language unambiguously presents a mathematical expression. In the continuation of this article, we study a family of expression-generating context-free grammars. Let us first provide a simple example and then continue with defining \emph{expressions}. 
\begin{ex}
    Consider a grammar $G=(\{S, M\}, \{\texttt{x}\}, S, R)$ with the rules 
    $$
    \begin{array}{rcl}
         S & \to & S \texttt{+} M \mid M, \\
         M & \to & M\texttt{x} \mid \texttt{x}.
    \end{array}
    $$
Every string $w \in L(G)$ is of the form
$$
w = \underbrace{\texttt{x}\cdots\texttt{x}}_{k_1}
    \texttt{+} \underbrace{\texttt{x}\cdots\texttt{x}}_{k_2}
    \texttt{+} \cdots
    \texttt{+} \underbrace{\texttt{x}\cdots\texttt{x}}_{k_j} 
$$
where $k_i\geq 1$. Thus, the strings in $L(G)$ correspond to polynomials $\sum \limits_{i=1}^{m} a_i x^i$ for some $m\in\N$ and $a_i\in \N_0$, where $\sum \limits_{i = 1}^m a_i > 0$.
\end{ex}

Expression-generating grammars are often used in equation discovery~\cite{todorovski1997, brence2021}, also known as symbolic regression. Symbolic regression study machine learning algorithms for training models that take a form of closed-form equations from data. The appropriate model is selected both in terms of accuracy on training data and its simplicity. 

Most equation discovery algorithms follow a general generate-and-test paradigm for training the model. In the first (generate) phase, the algorithm generates expressions, following some generative model. The latter is often a stochastic process. In evolutionary approaches to equation discovery~\cite{schmidt2009}, the expressions are generated following their fitness, i.e., degree of fit to the training data. In grammar-based approaches~\cite{brence2021}, PCFG is used to generate expressions, where PCFG can specify various aspects of inductive bias: the space of expressions considered~\cite{todorovski1997} and the preference towards simpler equations~\cite{brence2021}.

In the second (test) phase, the algorithm estimates the degree of fit of the generated expression to the training data. It also often takes care of fitting the values of free constants in the expression to data. For example, rather than generating \texttt{2 x\textsubscript{1} + 3 x\textsubscript{2}} and \texttt{3 x\textsubscript{1} + 4 x\textsubscript{2}}, one can use a PCFG that generates the string \texttt{c x\textsubscript{1} + c x\textsubscript{2}}, where \texttt{c} is a symbol denoting a generic free constant. After the expression is generated, it is post-processed, so that the $i$-th occurrence of the symbol \texttt{c} is replaced by \texttt{c\textsubscript{i}}. Following the example above, we obtain the final string \texttt{c\textsubscript{1} x\textsubscript{1} + c\textsubscript{2} x\textsubscript{2}}. Finally, the expression is input to the constant-fitting algorithm that finds the value for each of the constants that leads to maximal fit with the training data.

 \begin{ex}
     Observe a data set, sampled from the equation $y = 2.5 x_1 - x_2$:
\begin{center}
    \begin{tabular}{rr|r}
    $x_1$ & $x_2$ & $y$ \\
    \hline
1  & 4 & -1.5\\
2  & 7 & -2.0\\
1  & -8 & 10.5\\
6  & -10 & 25.0
\end{tabular}    
\end{center}
Suppose that the expression \texttt{c x\textsubscript{1}\textsuperscript{2} + c} is generated in the first try. After post-processing (\texttt{c\textsubscript{1} x\textsubscript{1}\textsuperscript{2} + c\textsubscript{2}}) and fitting the constants, we obtain the final equation $y = 0.65 x_1^2 + 1.2$ (if least-squares are used for fitting). In the second try, suppose the expression \texttt{c x\textsubscript{1} + c x\textsubscript{2}} is generated.
Then, it is post-processed to \texttt{c\textsubscript{1} x\textsubscript{1} + c\textsubscript{2} x\textsubscript{2}} and fitted to $y = 2.5 x_1 - 1.0 x_2$.

\end{ex}
Given a PCFG $G$, we want to compute the probability of generating a string that corresponds to a given expression $w$, taking into account all the strings in $L(G)$ that correspond to expressions that are equivalent to $w$. This is different from computing the probability of a $w$ as a string. For example, the strings \texttt{c x} and \texttt{c x + c x} are different, but the corresponding families of functions $\{x\mapsto c_1 x \mid c_1\in \mathbb{R} \}$ and $\{x\mapsto c_1 x_1 + c_2 x_2 \mid c_1, c_2\in\mathbb{R}\}$ are the same. More formally:


\begin{definition}
    Let $G$ be an equation generating PCFG. Suppose the set of terminals $T$ contains the symbols \texttt{x\textsubscript{1}}, \dots, \texttt{x\textsubscript{n}} that correspond to variables $x_1$, \dots, $x_n$. Let $D$ be their domain, and $\F$ be the domain for the free constants. Every string $w\in L(G)$ can we written as $w = w_1 \texttt{c} w_2 \texttt{c} \cdots w_{m} \texttt{c} w_{m + 1}$ for some $m\in\N_0$, where $\texttt{c} \notin w_i$, for $1\leq i\leq m + 1$. Given such a string and some values $c_1,\dots, c_m\in\F$, we define a function $f_{w, c_1, \dots, c_{m}}$ as
    \begin{align*}
        f_{w, c_1, \dots, c_{m }} &:\quad D_{w, c_1, \dots, c_{m }} &\to&&  K_{w, c_1, \dots, c_{m }}\hphantom{oooooooo} \\
        f_{w, c_1, \dots, c_{m }} &:\quad (x_1, \dots, x_n) & \mapsto && w_1 c_1 w_2 c_2 \dotsm w_m c_{m } w_{m + 1},    
    \end{align*}
    where the occurrences of \texttt{x\textsubscript{i}} on the right-hand side are replaced by $x_i$, $1\leq i\leq n$, the domain $D_{w, c_1,\dots c_{m}}$ is the largest possible (possibly empty) and $K_{w, c_1,\dots, c_{m }}$ equals the $f_{w, c_1, \dots, c_{m }}$-image of $D_{w, c_1, \dots, c_{m }}$.
    Let $\mathcal{F}$ be the set of all such functions.
    We define mapping $\Phi\colon L(G) \to 2^\mathcal{F}$ that maps $w\in L(G)$ to the set of functions
    \begin{align*}
        \Phi(w) = \{f_{w, c_1, \dots, c_{m }} \mid c_1, \dotsc, c_m \in \F \}
    \end{align*}
    that can be obtained from $w$ using different values for $m$ constants that appear in $w$. We define an equivalence relation $\sim$ on $L(G)$ as 
    $$
        w \sim v \Leftrightarrow \Phi(w) = \Phi(v).
    $$
    We define \emph{expressions} as the equivalence classes of the relation (elements of $L(G)/_\sim$) and denote them with $[w] = \{v \in L(G) \mid v \sim w\}$, $w \in L(G)$.
\end{definition}

The central aim of the next section of the article is to establish a general algorithm for calculating $P([w])$ for a given string $w$ and an expression-generating grammar $G$.

\section{Probability of generating an expression}
Let $G$ be an expression-generating grammar and $w \in L(G)$. We need an algorithm that calculates the probability of deriving any string $v$ that corresponds to an expression $[w]$, i.e.,
$$
P([w]) = \sum \limits_{v \sim w} P(v).
$$
From now on, we will use the short phrase ``deriving an expression $[w]$'' to replace the correct long version of ``deriving any string $v \sim w$ or $v \in [w]$''.

Let us first prove that the problem of calculating the probability is undecidable in general case, where an arbitrary expression-generating grammar can be provided at input. We can prove this by introducing an extension of a standard universal grammar for generating algebraic expressions, see, e.g.,~\cite{sipser2006}. Following the pattern of the universal grammar, the production rules for the starting non-terminal $E$ build sums of arbitrary number of factors $F$, non-terminal $F$ builds factors using multiplication and division of terms $T$. Terms $T$ can be simple variables $V$, generic (free) constants or constants with known values (e.g., $\pi$ or $\ln 2$) $C$, elementary functions of an expression $E$ gathered around the non-terminal $R$, and a simple bracketed expression $(E)$:
\begin{align*}
    E &\to \opcija{E \texttt{+} F}{p_{+}} \quad\mid\quad \opcija{E \texttt{-} F}{p_{-}} \quad\mid\quad \opcija{F}{1 - (p_{+} + p_{-})} \\
    F &\to \opcija{F \cdot T}{p_{\cdot}} \quad\mid\quad \opcija{F \texttt{/} T}{p_{/}} \quad\mid\quad \opcija{T}{1 - (p_{\cdot} + p_{/})} \\
    T &\to \opcija{\texttt{(}E\texttt{)}}{p_E} \quad\mid\quad \opcija{R}{p_R} \quad\mid\quad \opcija{V}{p_V} \quad\mid\quad \opcija{C}{1 - (p_E + p_R + p_V)} \\
    R &\to \opcija{\texttt{sin(}E\texttt{)}}{p_{\sin}} \quad\mid\quad \opcija{\texttt{exp(}E\texttt{)}}{p_{\exp}} \quad\mid\quad \opcija{\texttt{|}E\texttt{|}}{1 - (p_{\sin} + p_{\exp})} \\
    V &\to \opcija{\texttt{x}}{1} \\
    C &\to \opcija{\pi}{p_{\pi}} \quad\mid\quad \opcija{\ln 2}{p_{\ln}} \quad\mid\quad \opcija{Q}{p_Q} \quad\mid\quad \opcija{(\texttt{-}Q)}{1 - p_{\pi} - p_{\ln} - p_Q} \\
    Q &\to \opcija{\texttt{(}N\texttt{)} \texttt{/} \texttt{(}N \texttt{+} 1\texttt{)}}{1} \\
    N &\to \opcija{N \texttt{+} 1}{p_N} \quad\mid\quad \opcija{0}{1 - p_N}
\end{align*}

The selection of grammar rules corresponding to $R$ (elementary functions) and $C$ (constants, leading also to $Q$ and $N$) are aligned with the expressions in the Richardson's theorem used below to establish the undecidability of the general problem of calculating the probability of a given expression.

\subsection{Undecidaility for an arbitrary grammar}
\label{section:undeciability}
We will show that no general algorithm exists that would return $P([w])$ for every given grammar. Our proof relies on Richardson's theorem from~\cite{richardson}.
\begin{thm}[Richardson]
    Let $L$ be a set of strings representing $\R \to \R$ functions, containing strings $x$ (which represent identity), $\sin x, e^x, \ln 2, \pi, |x|$ and a set of rational numbers. Suppose that the set of functions represented by the strings in $L$ is closed under pointwise addition, subtraction, multiplication and composition. Then the problem of deciding, whenever a given string from $L$ represents a function, that is zero everywhere, is undecidable. 
\end{thm}
It is now easy to prove undecidability of calculation of $P([w])$. 
\begin{thm}
    \label{thm:undeciability}
    There is no general algorithm, that would take a grammar $G = (N, T, S, R)$ and a string $w \in T^*$, and would return the probability $P([w])$ of deriving an expression $[w]$.
\end{thm}
\begin{proof}
    Suppose such an algorithm exists for a grammar whose language $L(G)$ are strings corresponding to the functions that satisfy the assumptions of Richardson's theorem. For any $w\in L(G)$, we construct a probabilistic grammar $G_w = (\{S\}, T, S, R)$ with two rules 
    $$
         S \to \opcija{0}{0.5} \mid \opcija{v}{0.5}.
    $$
    Note that string $v$ represents a function that is zero everywhere, if and only if $P([v]) = 1$. Thus, if we are able to calculate the probability $P([v])$ for all $v\in L(G)$, we obtain an algorithm, that contradicts Richardson's theorem. Therefore, a general algorithm for computing $P([w])$ does not exist. 
\end{proof}

In the rest of this section, we focus on families of grammars, often used for equation discovery, for which such an algorithm exists, which is proved by computing the corresponding probabilities explicitly.

\subsection{Linear grammar}
\label{linear-grammar}
Observe a grammar with start symbol $E$, the set of terminal symbols $T = \{\texttt{c}, \texttt{+}, \texttt{x\textsubscript{1}}, \dotsc , \texttt{x\textsubscript{n}}\}$ and production rules
$$
\begin{array}{rcl}
    E & \to & \opcija{E \, \texttt{+} \, \texttt{c}V}{p} \mid \opcija{\texttt{c}}{1 -p}\\
    V & \to & \opcija{\texttt{x\textsubscript{1}}}{q_1} \mid \dotsb \mid \opcija{\texttt{x\textsubscript{n}}}{q_n}.
\end{array}
$$
Strings derived from this grammar take the form \texttt{c + c x\textsubscript{r\textsubscript{1}} + $\cdots$ + c x\textsubscript{r\textsubscript{k}}}, where $r_i \in \{1, \dotsc , n\}$. We are interested in $P([c +cx_{r_1} + \dotsb + cx_{r_k}])$.

\subsubsection{Exact formula}
We start with the easiest case of $n = 1$. If we use $E\to \texttt{c}$ at the beginning of the derivation, we will derive $\texttt{c}$. Otherwise, we rewrite $E\to E\; \texttt{+ c} V \to E\; \texttt{+ c x\textsubscript{1}}$, and continue recursively with rewriting $E$. Thus, the parse trees $\tau_i$, $i\geq 0$ derived with the grammar can be recursively defined as
\begin{center}
   \begin{tblr}{c c c c}
     $\tau_0 =$ & 
     \begin{forest} [$E$ [\texttt{c}]]] \end{forest} &
     \hphantom{SSS} and\hphantom{SSS} $\tau_{i + 1} =$ &
        \begin{forest}
         [$E$
         [$\tau_i$]
         [\texttt{+}]
         [\texttt{c}]
         [$V$ [\texttt{x\textsubscript{1}}]]
         ]
        \end{forest}
\end{tblr}
\end{center}
with the probability of parsing $\tau_i$ being equal to $(p q_1)^i (1-p)$, where $i$ is also the number of occurrences of symbol \texttt{x\textsubscript{1}} in the string.
Thus, every tree corresponds to a different string. Since $q_1 = 1$, we have
$$
P([c+cx_1]) = \sum \limits_{i=1}^\infty (1-p)p^i = (1-p)\frac{p}{1-p} = p.
$$
Since $P([c]) + P([c + cx_1]) = 1$, we also proved that the process of production terminates with probability $1$.



We can derive a formula for $n > 1$ using the inclusion-exclusion principle. Consider a string $w = \texttt{c + c x\textsubscript{r\textsubscript{1}} + $\cdots$ + c x\textsubscript{r\textsubscript{k}}} \in L(G)$, where $1\leq r_1 < \dotsb < r_k \leq n$. Observe that a given string is an element of $[w]$ if and only if it contains at least one symbol \texttt{x\textsubscript{i}} precisely when $i \in \{r_1, \dotsc, r_k\}$.
We define $U=\{v \in L(G) \mid \forall j \in \{1, \dotsc , n\}\setminus \{r_1, \dotsc, r_k\}.\; \texttt{x\textsubscript{j}} \notin v\}$. An expression $[w]$ can now be written as 
$$
[w] = \bigcap \limits_{i=1}^k \{v \in U \mid \texttt{x\textsubscript{r\textsubscript{i}}} \in v\}.
$$
The set $U$ is a disjoint union of  a set $[w]$ and a set 
$$[w]^c := \bigcup \limits_{i=1}^k\{ v \in U \mid \texttt{x\textsubscript{r\textsubscript{i}}} \notin v\}.$$
Therefore  $P([w]) = P(U) - P([w]^c)$. 

Let  $p_i$ denote the sum of probabilities of parsing all the words from $U$, which contain at most $i$ symbols \texttt{+}.
Clearly, $\lim \limits_{i \to \infty}p_i = P(U)$ and $p_0= 1 - p$.  
We derive a recursive formula for $p_i$. Every tree, which derives a string from $U$ and contains between $1$ and $i+1$ symbols $+$, $i \geq 0$, has the rule $E \to E\texttt{ + c}V$ in the root node. Then, the left subtree of the root node starts with $E$ and derives a word from $U$ with at most $i$ occurrences of symbol \texttt{+}. The probability of such subtree is $p_i$. The right subtree that starts with $V$ can be rewritten to one of the symbols from  $\{\texttt{x\textsubscript{r\textsubscript{1}}}, \dots, \texttt{x\textsubscript{r\textsubscript{k}}}\}$. This happens with the probability $\sum \limits_{j=1}^k q_{r_j}$. Therefore,
$$
p_{i+1} = (1 - p) + pp_i \sum \limits_{j=1}^k q_{r_j},
$$
and, in the limit, $P(U) = (1 - p) + p P(U) \sum \limits_{j=1}^k q_{r_j}$, so
$$P(U) =\displaystyle \frac{1-p}{1 -p \sum \limits_{j=1}^k q_{r_j}}.$$
We get $P([w]^c)$ with inclusion-exclusion principle as  
$$
P(\bigcup \limits_{i=1}^k\{ v \in U \mid \texttt{x\textsubscript{r\textsubscript{i}}} \notin v\}) = 
\sum \limits_{\emptyset \neq I \subseteq \{1, \dotsc, k\}} (-1)^{|I| + 1} P(A_I),
$$
where $A_I = \bigcap \limits_{i \in I} \{v \in U \mid \texttt{x\textsubscript{r\textsubscript{i}}} \notin v \}$. The derivation of the formula for $P(A_I)$ is similar to the derivation of the formula for $P(U)$ described above. Let $p_i'$ be the probability of parsing a word from $[w]^c$, which contains at most $i$ symbols $+$. Then $P(A) = \lim \limits_{i\to \infty} p_i'$ and (as derived above)
$$
p_{i+1}' = (1-p) + pp_i'\sum \limits_{j \in \{1, \dotsc, k\} \setminus I}q_{r_j},
$$
from which  $P(A_I) = \displaystyle \frac{1-p}{1-p\sum \limits_{i \in \{1, \dotsc, k\} \setminus I}q_{r_i}}$ follows.
Therefore
\begin{equation}
    \label{eq:linearna}
    P([w]) = \sum \limits_{I \subseteq \{1, \dotsc, k\}} (-1)^{|I|} \displaystyle \frac{1-p}{1-p\sum \limits_{i \in \{1, \dotsc, k\} \setminus I}q_{r_i}}.
\end{equation}
For the calculation of the sum \eqref{eq:linearna}, we need 
$$
 \sum \limits_{I \subseteq \{1, \dotsc, k\}}  ( |\{1, \dotsc, k\} \setminus I| + 5) = \sum \limits_{i=0}^k  \binom{k}{i}(i+ 5) = 2^{k-1} (10 + k)
$$
elementary computing operations, where $k$ is the number of different variables  $x_i$ in the string $w$. This leads to an exponential time complexity with respect to $k$. For practical use of the formula, we need an efficient approximation of the exact formula.

\subsubsection{Approximation of the exact formula}

Let $w$ be as above a string that contains the symbols \texttt{x\textsubscript{r\textsubscript{1}}}, \dots \texttt{x\textsubscript{r\textsubscript{k}}}, and no other symbols \texttt{x\textsubscript{j}}.
Let $v \in [w]$ be a string with $i$ occurrences of a symbol \texttt{+} and let $l_j$ be the number of occurrences of a symbol \texttt{x\textsubscript{r\textsubscript{j}}} in the string $v$. Clearly, $l_1 + \dotsb + l_k = i$ and $l_j \geq 1$. For each such valid tuple $(l_1, \dotsc, l_k)$, there exist  
$\binom{i}{l_1,\, \dotsc,\, l_k}$ different strings in $[w]$, which fulfill the upper demands. Since the probability of parsing  one such word is equal to $(1-p) p^i q_{r_1}^{l_1} \dotsm q_{r_k}^{l_k}$, the probability of parsing  $[w]$ is equal to 
\begin{equation}
\label{eq:eq:linearna-vsota-neskoncna}
P([w])=
\sum \limits_{i= k}^\infty (1-p)p^i  \left ( \sum_{\substack{l_1 + \dotsb + l_k = i \\ l_j \geq 1}}\binom{i}{l_1, \dotsc, l_k} q_{r_1}^{l_1} \dotsm q_{r_k}^{l_k} \right )    .
\end{equation}
\begin{figure}[t]
\begin{subfigure}{0.49\textwidth}
    \centering
    \includegraphics[scale=0.25]{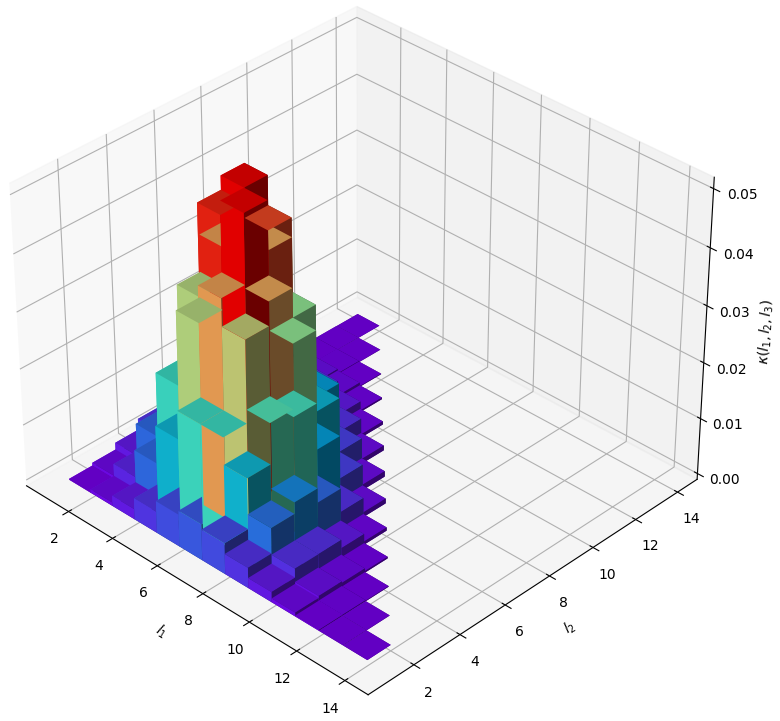}
\end{subfigure}
\begin{subfigure}{0.49\textwidth}
    \centering
    \includegraphics[scale=0.25]{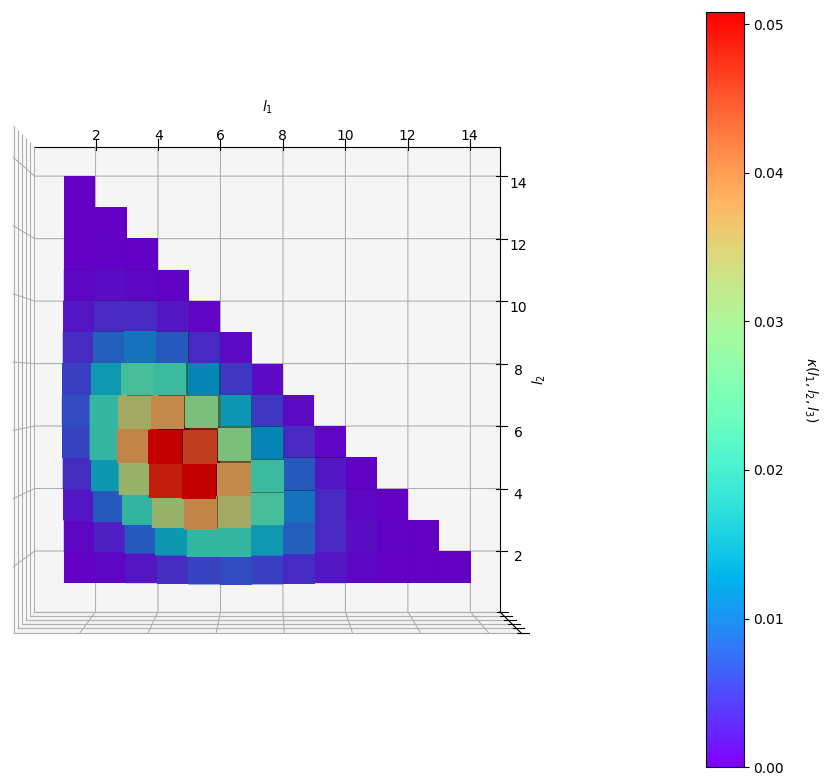}
\end{subfigure}    
\caption{3D-bar plot (in two perspectives) of the sizes of the terms $\Tilde{\kappa}$ in the inner sum of Eq.~\eqref{eq:double-sum}, for $i = 15$, $k =3$, and $q_1 = q_2 = q_3 = 0.3$.}
\label{fig:multinomial_distribution}
\end{figure}

A good approximation for $P([w])$ would be
\begin{equation}
    \label{eq:double-sum}
\sum \limits_{i=k}^M (1-p)p^i  \left ( \sum_{\substack{l_1 + \dotsb + l_k = i \\ l_j \geq 1}}\M \right )    
\end{equation}
where 
$$
\M = \binom{i}{l_1, \dotsc, l_k} q_{r_1}^{l_1} \dotsm q_{r_k}^{l_k}
$$
and  $M\in \N$, but the calculation using this formula is still slow, due to the high number of terms in the inner sum -- for a given $i$, their number is $\binom{i-1}{k-1}$. 
        Observe the  graphs  of $\Tilde{\kappa} \colon (l_1, \dotsc, l_{k-1}) \mapsto \kappa(l_1, \dotsc, l_{k-1}, i - \sum_{j = 1}^{k - 1} l_j)$ in Figure~\ref{fig:multinomial_distribution}. Clearly, $\kappa$ is a generalization of probability mass function of multinomial distribution, so there is (at least one) partition $p'= (l_1', \dotsc, l_{k-1}', l_k')$, $l_k' = i -\sum_{j = 1}^{k - 1} l_j'$, where $\Tilde{\kappa}$ reaches its maximum and the terms $\kappa(l_1, \dotsc, l_k)$ in the inner sum of Eq.~\eqref{eq:double-sum} decrease when the distance between $(l_1 , \dotsc, l_{k-1})$ and $(l_1' , \dotsc, l_{k-1}')$ increases. Therefore, the partitions that are far away from $p'$ can be ignored in the approximation.

\begin{figure}[t]
\begin{subfigure}{0.49\textwidth}
    \centering
    \includegraphics[width=\textwidth]{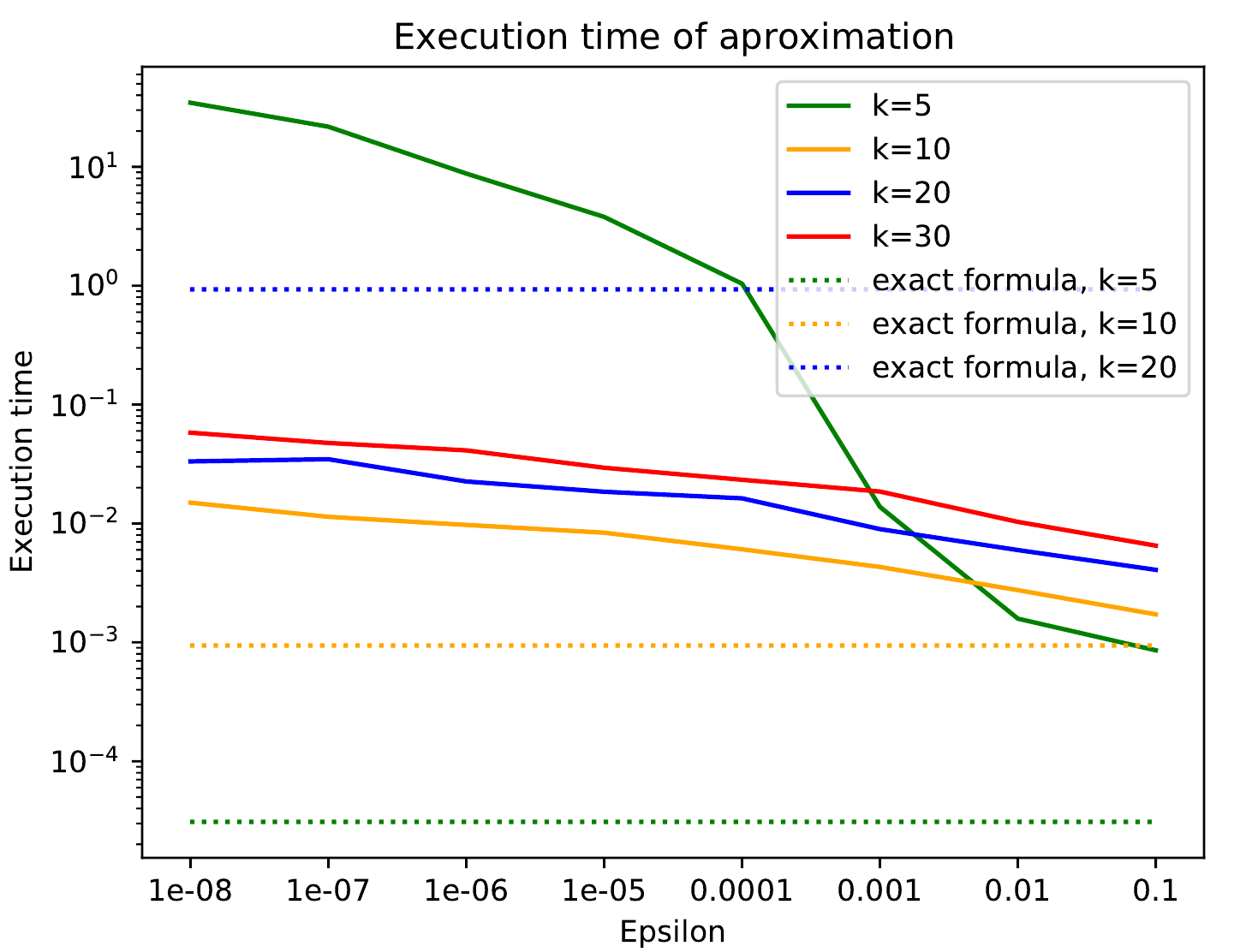}
   \caption{Execution time (in seconds), for $p=0.5$.}
\end{subfigure}
\begin{subfigure}{0.49\textwidth}
    \centering
    \includegraphics[width=0.96\textwidth]{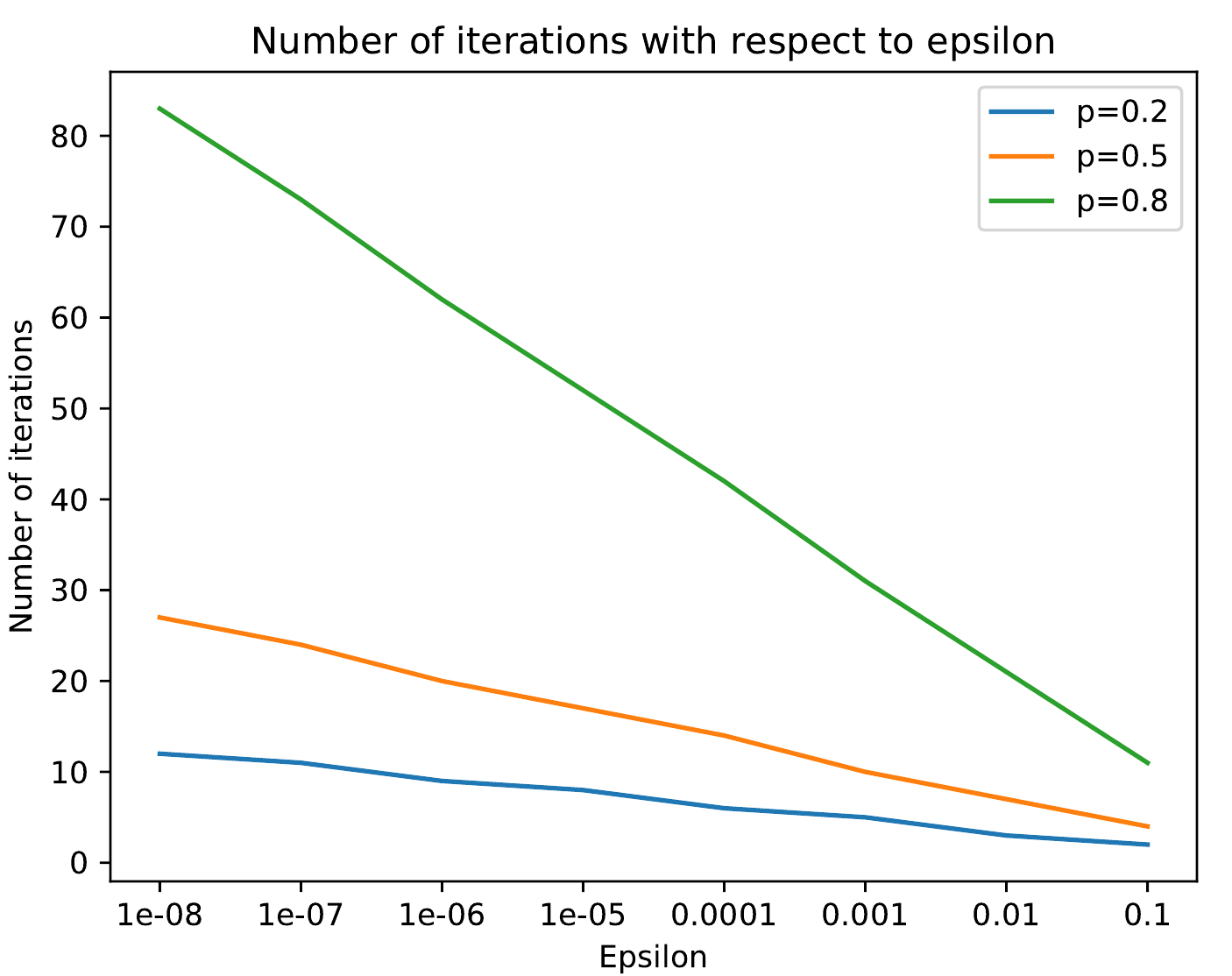}
    \caption{Number of computed terms, for $k = 10$.}
\end{subfigure}    
\caption{Computational complexity of the exact calculation and the approximation of the probability of the expression $c + c x_1 + \cdots + c x_{k}$ for the grammar with the rules
$E \to \opcija{E \, \texttt{+} \, \texttt{c}V}{p} \mid \opcija{\texttt{c}}{1 -p}$,
and
$V \to \opcija{\texttt{x\textsubscript{1}}}{q_1} \mid \dotsb \mid\opcija{\texttt{x\textsubscript{n}}}{q_k}$.
r different values of the desired accuracy $\epsilon$ and different numbers of variables $k$, where $q_1 = \dotsb = q_k = 1/k$. Graph (a) depicts the execution time (in seconds), and graph (b) depicts the number of iterations required to achieve the desired accuracy $\epsilon$.}
\label{graphs-time}
\end{figure}

We approximate each inner sum by starting with partitions that correspond to the (approximate) maximum of $\Tilde{\kappa}$ and moving to smaller elements using breadth-first search until enough elements are included to reach the desired precision. The number of iterations $M - k + 1$ and the number of ignored elements in each iteration can both be set to reach an approximation error smaller than $\epsilon$ for arbitrary given $\epsilon > 0$. \ref{apendix:izpeljava-aprox} provides implementation-level detailed description of the approximation algorithm. Let $\epsilon' = \epsilon / (2(M - k + 1))$ and $\overline{m}_k$ be the value of $\M$ in the mode at $i = k$. As derived in \ref{apendix:aprox-complexity}, time complexity of the algorithm is bounded by
$$
T(M) \leq \bigO(M^2)\left( (M - k + 1) \binom{M}{k-1} - \frac{\epsilon'(1 - p^{M - 1 - k})}{(1 - p)^2 p^{M - 2} \overline{m}_k} \right).
$$

Figure~\ref{graphs-time} shows the utility of the approximation by showing the empirical comparison of the complexity of the exact algorithm to the approximate one for different values of $k$ and varying desired precision $\epsilon$. The comparison shows that for values of $k > 20 $, the approximate algorithm is significantly faster than the exact one.

\subsection{Polynomial grammar}
\label{polynomial-grammar}
Linear grammar from Section~\ref{linear-grammar} can be adapted to derive polynomials with constants. Instead of rewriting $V$ to a symbol for single variable, we rewrite it to a sequence of symbols that corresponds to a monomial $x_1^{m_1}\dotsb x_n^{m_n}$. Hence, we get the following \emph{polynomial grammar}, which generates polynomials  $c + cx_1^{m_{1, 1}}\dotsm x_n^{m_{n,1}} + \dotsb + cx_1^{m_{1, k}}\dotsm x_n^{m_{n,k}}$:
$$
\begin{array}{rcl}
    E & \to & \opcija{E \, \texttt{+} \, \texttt{c}V}{p} \mid \opcija{\texttt{c}}{1 -p}\\
    V & \to & \opcija{VF}{q} \mid \opcija{F}{1-q} \\
    F & \to & \opcija{\texttt{x\textsubscript{1}}}{q_1} \mid \dotsb \mid \opcija{\texttt{x\textsubscript{n}}}{q_n}.
\end{array}
$$
Note that $x_1^2 x_2$ cannot appear as a string in the language of this grammar, but only as an expression that corresponds to the strings \texttt{x\textsubscript{1}x\textsubscript{1}x\textsubscript{2}}, \texttt{x\textsubscript{1}x\textsubscript{2}x\textsubscript{1}} and \texttt{x\textsubscript{2}x\textsubscript{1}x\textsubscript{1}}. However, for readability, we abbreviate these strings as \texttt{x\textsubscript{1}\textsuperscript{2} x\textsubscript{2}}, and do similarly for the others in these section.

We now show that the computation of probability 
$$
P([\texttt{c + c x\textsubscript{1}\textsuperscript{m\textsubscript{1,1}}} \cdots \texttt{x\textsubscript{n}\textsuperscript{m\textsubscript{n,1}}} 
+ \dotsb + 
\texttt{c x\textsubscript{1}\textsuperscript{m\textsubscript{1,k}}}\cdots \texttt{x\textsubscript{1}\textsuperscript{m\textsubscript{n,k}}}])
$$
for a given polynomial grammar can be translated to a problem of calculating the probability for a linear grammar. For start, note that every branch
$$
V \to VF \to VFF \to \dotsb \to VF\dotsm F \to FF\dotsb F \to 
\texttt{x\textsubscript{r\textsubscript{1}}} \dotsb \texttt{x\textsubscript{r\textsubscript{k}}}$$
in the parse tree can be replaced with a single production rule $V \stackrel{p}{\to} x_1^{m_1}\cdots x_n^{m_n}$, where $x_1^{m_1}\cdots x_n^{m_n}$ should be understood as a single symbol and $p$ is the probability of generating the expression $x_1^{m_1}\cdots x_n^{m_n}$ from $V$ with a polynomial grammar. The resulting grammar is linear, so the procedure from Section~\ref{linear-grammar} can be applied to this grammar as well.

\begin{note}
The obtained structure is formally not a grammar, since it includes an infinite number of production rules for $V$, and an infinite number of terminals. However, for the purpose of computing the probabilities of expressions, we can allow any countable set of production rules and generalize the previous definitions (and procedures).
\end{note}

\noindent For a full algorithmic solution, the probability of deriving the expression $x_1^{m_1}\dotsm x_n^{m_n}$ from  $V$ must be calculated.
Let $M = m_1 + \dots + m_n$. To derive the expression from $V$, the rule $V\to VF$ must be applied $(M - 1)$-times and the rule $V\to F$ once. Then, each of the rules $F\to \texttt{x\textsubscript{i}}$ must be applied $m_i$ times. Thus, the probability of deriving the expression from  $V$ equals
$$
p = \binom{M}{m_1, \dotsc , m_n}q^{M-1}(1-q)q_1^{m_1}\dotsb q_n^{m_n}.
$$

\subsection{Rational grammar}
\label{rational-grammar}

Using the polynomial grammar from the previous section, we now define a grammar for deriving \emph{rational functions}. We introduce a new start symbol $S$ that is rewritten to the quotient $(E)/(E)$. Then, each $E$ is rewritten to an arbitrary polynomial:
$$
\begin{array}{rcl}
    S & \to & \opcija{(E) \texttt{/} (E)}{1}\\
    E & \to & \opcija{E\texttt{ + c}V}{p} \mid \opcija{\texttt{c}}{1 -p}\\
    V & \to & \opcija{VF}{q} \mid \opcija{F}{1-q}\\
    F & \to & \opcija{\texttt{x\textsubscript{1}}}{q_1} \mid \dotsb \mid \opcija{\texttt{x\textsubscript{n}}}{q_n} 
\end{array}
$$
Strings $w$ derived by this grammar take the form of $u/v$ for any two strings $u$ and $v$ from the language of the polynomial grammar. To understand the structure of $[u/v]$, we use the following lemma.
\begin{lema}
\label{lema:racional-structure}
Let $u/v$ and $s/t$ be two strings from the language of the rational grammar, such that $u/v \sim s/t$. Than, $u \sim s$ and $v \sim t$.
\end{lema}
The immediate consequence of the lemma is that $P([u/v]) = P([u])P([v])$. The probabilities $P([u])$ and $P([v])$ can be calculated using the algorithm from Section~\ref{polynomial-grammar}.

\begin{proof}[Proof of Lemma \ref{lema:racional-structure}.]
First, we will prove $u \sim s$. 
To do so, it is sufficient to prove that the same monomials occur in $u$ and $s$.

Let $N = x_1^{l_1}\dotsm  x_n^{l_n}$ be a monomial of the highest degree in $u$. We set the constant next to $N$ to $1$ while all the other free constants in $u$ are set to $0$.  In $v$ we set the free term (the constant next to the monomial $1$) to $1$, and all the other constants are set to $0$. This selection of constants in $u/v$ is equal to $N$ as a function, and since $u/v \sim s/t$, there exist a selection of constants $s'/t'$ for $s/t$ (where $s'$ is a polynomial obtained with a selection of constants in $s$, and $t'$ is a polynomial obtained with a selection of constants in $t$), such that $N$ is equal to $s' /t'$ (as a function).
Clearly, a degree  $\deg(s')$ must be equal to or larger than the degree $\deg (N)$. Since $N$ has the maximal degree in $u$, we have $\deg (s) \geq \deg(u)$. By a symmetric argument, we can prove that $\deg (u) \geq \deg(s)$, so the degrees of $u$ and $s$ are equal.

Let $u'' = N + \dotsb$ be a polynomial, which we get from a selection of constants in $u/v$, such that all the free constants in $u$ are set to $1$, the free term in $v$ is set to $1$, and all the other constants in $v$ are set to $0$. There exists selection of constants $s''$ in $s$ and $t''$ in $t$ that leads to $u'' = s'' / t''$, or 
$$
s'' = t'' u'' = t'' (N + \dotsb),
$$
where $t''$ is not the zero polynomial.
As proved above, $\deg(s'') = \deg(N)$. Since $\deg(s'') = \deg(t'') + \deg(N)$, we have $\deg(t'') = 0$, i.e., $t''$ is a constant polynomial.
Since $s'' = t'' u''$, $s$ must include all the monomials from $u$. After appplying a symmetric argument, we prove that $s$ and $u$ contain the same monomials, so $u \sim s$. Following the same reasoning, we can prove $v \sim t$.
\end{proof}

\subsection{A note on alternative grammars for linear expressions}
The linear grammar, presented in section \ref{linear-grammar}, is used in algorithms such as~\cite{brence2021}, due to easy interpretation of probabilities $P(V  \to \texttt{x\textsubscript{i}})$. But any other grammar, that generates a subset of linear expressions with constants, could be used as well. Here, we present an alternative with a simpler formula for probabilities.
$$
\begin{array}{lcl}
    S & \to & \opcija{V_1 \texttt{ + c}}{p_0}\;\mid\; \opcija{\texttt{c}}{1 -p_0}\\
    V_1 & \to & \opcija{V_2 \texttt{ + c x\textsubscript{1}}}{p_1} \;\mid\; \opcija{V_2}{q_1} \;\mid\; \opcija{\texttt{c x\textsubscript{1}}}{1-p_1 - q_1} \\
    \vdots \\
    V_{n-1} & \to & \opcija{V_n \texttt{ + c x\textsubscript{n-1}}}{p_{n-1}} \;\mid\; \opcija{V_n}{q_{n-1}} \;\mid\; \opcija{\texttt{c x\textsubscript{n-1}}}{1-p_{n-1} - q_{n-1}} \\
    V_n  &\to & \opcija{\texttt{c x\textsubscript{n}}}{1}
\end{array}
$$
The equivalence classes $[w]$ for words $w$ of this grammar all contain only a single word (unlike the equivalence classes of the grammar from Section \ref{linear-grammar}). Thus, $P(w) = P([w])$ for any word $w$ and the probability of parsing the expression $[w] = c +cx_{r_1} + \dotsb + cx_{r_k}$ is equal to 
$$
P([w]) =
\begin{cases}
    1 - p_0&;\; k = 0\\
    p_0\prod \limits_{i = 1}^{M }g(i)&;\; \text{otherwise}
\end{cases}\text{ ,\hphantom{ooooooiiooooo}}
$$
where  $M=\operatorname{max}(r_1, \dotsc, r_k)$ and
$$g(i)  = 
    \begin{cases}
        p_i &  i \in \{r_1, \dotsc, r_k\} \setminus\{M\}  \\
        1-p_i-q_i &  r_i = M \\
        q_i &\text{otherwise}
    \end{cases}.
$$

\section{Conclusion}

The article focuses on expression-generating probabilistic context-free grammars used for symbolic regression. We define expressions as equivalence classes of strings derived by grammar. We show that the problem of calculating the probability of deriving the strings in an equivalence class is undecidable in the general case of universal grammar for algebraic expressions. We present an algorithm for calculating the probability of a given expression generated with a given grammar for linear, polynomial, and rational expressions. Finally, we show that the exact probability can be efficiently approximated to a specified precision.

Two venues for further research emerge. First, a relevant open question for symbolic regression is what would be a most general restriction of the »Richardson« universal grammar presented in Section~4 that would allow for an algorithmic solution for calculating the probability of a given expression. For example, designing such an algorithm for a grammar that generates arbitrary algebraic expressions with four standard operators and generic constants would benefit symbolic regression. Second, the presented results on expression-generating expressions can be generalized to grammars generating groups, which might prove helpful in developing generative models and machine learning methods for algebraic structures.

\section*{Acknowledgements} 
The authors acknowledge the financial support of the Slovenian Research Agency via the research core funding No.~P2-0103 and No.~P1-0294 as well as project No.~N2-0128.
%
%

\appendix

\section{Removing linear cycles from probabilistic context-free grammars}
\label{app:removing-linear-cycles}
In this part, we assume that no null rules (rules of the form $A\to \epsilon$) are present in the grammar.
\begin{thm}
\label{thm:cycle-removal}
Let $G$ be a grammar and $ A_1 \xrightarrow{p_1} A_2 \xrightarrow{p_2} \dotsb A_{m-1} $ $ \xrightarrow{p_{m-1}}
        A_m \xrightarrow{p_m} A_1$ the longest linear cycle of pairwise distinct symbols $A_1$, $\dotsc$, $A_m$. Production rules of $G$ can be transformed so that a cycle $A_1 \to A_2 \to \dotsb \to A_m \to A_1$ is removed from the grammar, no new cycles of length at least $m$ emerge and the language $L(G)$ with its probability distribution over words stays unchanged. 
\end{thm}
%
%
%
\begin{proof}
Let us first prove this for $m=1$. Observe a cycle  $A \xrightarrow{p} A$. Let  $\mathcal{A} = \{ A \to \alpha_i\; [c_i] \mid 1 \leq i \leq k \}$ be the set of production rules for symbol $A$, excluding the rule $A \to A$. 
 If $\mathcal{A}$ is the empty set (i.e., $p = 1$)
we do not do anything (note that a sequence of rules containing such a rule never derives a string). Otherwise, we
remove the rule $A \to A$ from the grammar and redefine the probabilities of other rules as $P(A \to \alpha_i) = c_i / (1-p)$,
so that they sum up to $1$. Clearly, the language $L(G)$ stays the same and no new linear cycles are created.
       
Assume $m>1$. Let $\mathcal{A}_1 = \{ A_1 \to \alpha_i\; [c_i] \mid 1 \leq i \leq k_1 \}$ be the set of production rules for symbol $A_1$, and let $\mathcal{A}_m = \{ A_m \to \beta_j [q_j] \mid  1 \leq j \leq k_m\}$
 be the set of production rules for symbol $A_m$, excluding rules $A_m \to A_1$,
 and $A_m \to \alpha_i$, for all $1\leq i\leq k_1$ (if any of these rules exist). We also define
$\tilde{p}_i = P(A_m \to \alpha_i)$.

We remove the rule $A_m \to A_1$ from the grammar and add new production rules for the symbol $A_m$ (if such rules already exist, we redefine them with new probabilities):
$$
\begin{array}{rcl}
     A_m & \to & \opcija{\alpha_i}{p_mc_i  + \tilde{p}_i}, \quad 1\leq i\leq k_1.
\end{array}
$$
By doing so, we simulate applying rule $A_m \to A_1$, followed by one of the rules for $A_1$. Observe that the set $L(G)$ and probability distribution over $L(G)$ are invariant to the transformation and stay the same. 

At least the rule $A_m\to A_2$ is now present in $R$, so the relation $R$ is still left-total. Now, we check that the sum of probabilities of the production rules for $A_m$ equals $1$. Before the transformation, we had $p_m +
\sum \limits_{j=1}^{k_m} q_j + \sum \limits_{i=1}^{k_1}\tilde{p}_i = 1 =\sum \limits_{i=1}^{k_1} c_i$. The sum of probabilities of production rules for $A_m$ after the transformation is equal to 
\begin{align*}
 \sum \limits_{j=1}^{k_m} q_j  + 
 \sum \limits_{i=1}^{k_1}( p_m c_i + \tilde{p}_i)  &= 
 p_m(
 \sum \limits_{i=1}^{k_1} c_i) + \sum \limits_{j=1}^{k_m} q_j  + 
 \sum \limits_{i=1}^{k_1}\tilde{p}_i 
 \\
 &= p_m + (1-p_m) \\
 &= 1,
\end{align*}
so the new rules are well-defined. 

Transformation creates a cycle $A_2 \to \dotsb \to A_m \to A_2$, which is of length $m-1$. Other linear cycles can also emerge due to some of the rules $A_m \to \alpha_i$, $1\leq i \leq k_1$. 
In that case, $\alpha_i = B$ must be a non-terminal symbol, such that $B \to \dotsb \to A_m$ and the transformation creates a cycle $A_m \to B \to \dotsb \to A_m$. In that case, 
$A_m \to A_1 \to B \to \dotsb \to A_m$
is a linear cycle that was present also before the transformation. Thus, its length is at most $m$, so the length of $A_m \to A_1 \to B \to \dotsb \to A_m$ is at most $m - 1$.
\end{proof}
An immediate corollary follows:
\begin{cor}
All the linear cycles of a given PCFG (that can derive a string) can be algorithmically removed with a transformation, that preserves the language and the probability distribution over the language.
\end{cor}
\begin{proof}
    If we apply the transformation from Theorem \ref{thm:cycle-removal} to the longest cycle, the number of cycles with the same length strictly decreases. Thus, we can remove all the linear cycles in a finite number of steps.
    The only exception are non-terminals $A\in N$, such that the rule $A\to A$ is the only rule for $A$. However, using such a rule can never lead to deriving a string.
\end{proof}
\begin{thm}
    Let $G = (N, T, S, R)$ be a grammar without linear cycles and null rules $A \to \epsilon$ and let $w \in L(G)$. There are finitely many parse trees that parse $w$.
\end{thm}
\begin{proof}
Denote the length of $w$ with $l$. There can be at most $l-1$ rules rewriting a symbol into at least two symbols in the tree. Other rules can be of a form $A \to B$ for some $A, B \in N$. There can be at most $|N|$ consecutive rules $A_1 \to \dotsb \to A_n$ (otherwise we would get a linear cycle. Therefore there are only finitely many possibilities for parse trees.
\end{proof}

\section{Probability-related (counter)examples}
\label{sec:appendix:examples}

In this part, we give an example of a grammar whose generating process does not finish with a positive probability, and an example of grammar where the probability of $P(\epsilon)$ cannot be expressed with radicals.

\begin{ex}
\label{ex:sum-of-probabilities}
The probabilities of parsing trees will not always sum to $1$.
We follow an example from~\cite{chi} and consider a PCFG with start symbol $S$, terminal symbol \texttt{x} and two production rules $S \to SS\ [p] \mid \texttt{x}\ [1-p]$. We denote with $p_i$ the sum of probabilities of all parse trees with depth at most $i$. Since $S \to \texttt{x}$ is the only possible parse tree with depth $1$, we have $p_1 = 1-p$.

Observe that each tree of depth at most $i \geq 2$ starts with $S \to SS$ (otherwise it would be equal to $S \to \texttt{x}$, and have a depth of $1$). Since the depth of  such a tree is at most $i$, the two subtrees, growing from the bottom of $S \to SS$, can stretch at most $i-1$ in depth. One such subtree occurs with a probability of $p_{i-1}$. From that, we derive a recursive formula 
$$
p_{i+1} = pp_i^2+1-p.
$$
Probabilities of all parse trees of the grammar $G$ sum up to $p_{\infty} := \lim_{i \to \infty} p_i$. From the recursive formula, we get  $p_{\infty} = pp_{\infty}^2 + 1-p$, from which we can derive $p_{\infty} = \min(1, \frac{1}{p} - 1)$. If  $p >  \frac{1}{2}$, the probability of parsing any parse tree is strictly smaller than $1$.
\end{ex}

\begin{thm}
\label{thm:p-epsilon-polynomial-zeros}
Computation of $P(\epsilon)$ for an arbitrary grammar is at least as hard as finding the roots of an arbitrary polynomial over $ \R_{\geq0}$. 
\end{thm}
\begin{proof}
    Consider a grammar with start symbol $A$, terminal symbol \texttt{x}, and rules 
    $$
    \begin{array}{rcl}
         A &\to& \opcija{\epsilon}{p_0} \\
         A &\to& \opcija{A}{p_1} \\
         A &\to& \opcija{AA}{p_2} \\
         \vdots \\ 
         A &\to& \opcija{\underbrace{A\dotsb A}_{n}}{p_n} \\
         A &\to& \opcija{\texttt{x}}{1-p_0 - \dotsb - p_n}
    \end{array}
    $$
    We want to compute $P(\epsilon)$. The first rule in any parse tree $\tau$ for which $\operatorname{str}(\tau) = \varepsilon$, must be one of the first $n + 1$ rules above. Let say we applied the $i$-th rule, $0\leq i\leq n$, which rewrites a single $A$ to the string that contains $i$ copies of $A$. If the final string is $\epsilon$, each of the copies should be rewritten to $\epsilon$ (which happens with the probability $P(\epsilon)^i$). Therefore,
    $$
    P(\epsilon) = \sum \limits_{i=0}^n p_i P(\epsilon)^i.
    $$
    $P(\epsilon)$ is one of the (real) roots of the polynomial $r(t) = -t + \sum \limits_{i=0}^n p_i t^i $. Since the probabilities $p_i$ and the degree of the polinomial are arbitrary, the exact computation of $P(\epsilon)$ is not possible.
    \end{proof}

\section{Approximation}
\label{apendix:izpeljava-aprox}
Recall that we want to approximate the formula
\begin{equation}
    \label{eqn:linar-repeated}
    P([w]) = \sum \limits_{i=k}^\infty (1-p)p^i  \left ( \sum \limits_{(l_1, \dotsc, l_k) \in \operatorname{Par}(i)}\M \right ) ,
\end{equation}
where $\operatorname{Par}(i) = \{(l_1, \dotsc, l_k) \mid l_1 + \dotsc + l_k = i \land l_j \in \N\}$ is the set of integer partitions of $i$. We do this in two steps. First, we skip the tail of the outer sum and compute only the terms for $i\in \{k, k + 1, \dots, M\}$, for some chosen $M$. This results in $\operatorname{Error}_M$. Second, in the inner sum we take a similar approach and compute the terms $\M$ only for partitions in $S(i) \subset \operatorname{Par}(i)$, where $S(i)$ will be defined later. This results $\operatorname{Error}_S(M)$. Thus, our approximation is
\begin{equation}
    \label{eq:linearna-priblizek}
\widehat{P(M, S)}=
\sum \limits_{i=k}^M (1-p)p^i  \left ( \sum \limits_{(l_1, \dotsc, l_k) \in S(i)}\M \right )    .
\end{equation}
We show that we can control both errors, so that for any $\epsilon> 0$, we can choose $M$ and sets $S(i)$, such that $\operatorname{Error}_M\leq \epsilon / 2$ and $\operatorname{Error}_S(M)\leq \epsilon / 2$. In this case, the total error of the approximation will be at most $\epsilon$.

\subsection{Number of iterations $M-k$}
Let $Q = q_1 + \dotsb + q_k$. If we compute only the first $M-k$ elements of the sum \eqref{eqn:linar-repeated}, we upper-bound the error
$\operatorname{Error_M} = P([w]) - \widehat{P(M, \operatorname{Par})}$ as follows:
\begin{align*}
\operatorname{Error_M} &= \sum \limits_{i=M+1}^{\infty} (1-p)p^i \sum 
\limits_{(l_1, \dotsc, l_k) \in \operatorname{Par}(i)} \kappa(l_1, \dotsc, l_k) \\
&\leq \sum \limits_{i=M+1}^{\infty} (1-p)p^i \sum 
\limits_{l_1+ \dotsb+ l_k = i} \kappa(l_1, \dotsc, l_k) \\
&\stackrel{(*)}{\leq} \sum \limits_{i=M+1}^{\infty} (1-p)(pQ)^i \\
&= (1-p) \frac{(pQ)^{M+1}}{1-pQ} = \operatorname{Error}_M^{'}.
\end{align*}
On the step $(*)$, we used multinomial theorem for $(q_1 + \cdots + q_k)^i$.
For any  $\epsilon>0$, we can find $M\in\mathbb{N}$, such that $\operatorname{Error}_M^{'}\leq \epsilon / 2$. An appropriate value of $M$ would be
$$
  M = \left\lfloor\frac{\log \left( \displaystyle \frac{\epsilon}{2}\cdot \displaystyle\frac{1-pQ}{1-p}\right)}{\log(pQ)}\right\rfloor
$$

\subsection{Construction of $S(i)$}
Choose $\epsilon > 0$ and let us assume we already selected the number of iterations $M$, such that $\operatorname{Error}_M\leq \epsilon / 2$.
Let 
$$E(i) = \sum \limits_{ (l_1, \dotsc, l_k) \in \operatorname{Par}(i) \setminus S(i)} \M$$
be the error, made by skipping some terms of the inner sum in \eqref{eqn:linar-repeated} for some $i$. The total error caused by this is
$\operatorname{Error_S}(M) = \sum \limits_{i=k}^M (1-p) p^i E(i)
$. We will construct $S(i)$, such that the error of each term $(1-p)p^i E(i)$ will not be greater than $\epsilon' = \epsilon / (2(M - k +1))$ and therefore $\operatorname{Error_S}(M)  \leq \epsilon /2 $.
We propose $S(i)$ to be the set of the partitions
that yield the $\lceil (1-\gamma _i) \binom{i-1}{k-1} \rceil$ highest values of $\M$. 
Since $|\operatorname{Par}(i) \setminus S(i)| / |\operatorname{Par}(i)|\leq \gamma_i$, the following  estimates can be made:
\begin{align*}
E(i) &\leq \gamma_i\binom{i-1}{k-1}\min \limits_{(l_1, \dotsc, l_k) \in S(i)} \M \\
     &\leq \gamma_i\binom{i-1}{k-1}\max \limits_{(l_1, \dotsc, l_k) \in S(k)} \M.
\end{align*}
The last step follows from the fact that the maximal probability $\M = \binom{i}{l_1, \dotsc, l_k} q_{r_1}^{l_1} \dotsm q_{r_k}^{l_k}$ at $i = k$
is greater that any other probability at $i > k$.
Let us denote $\overline{m}_k = \max\limits_{(l_1, \dotsc, l_k) \in S(k)} \M$. Then, we should have
$$
\gamma_i \leq \frac{\epsilon'}{(1 - p)p^i \overline{m}_k {i - 1 \choose k - 1}},
$$
so the values $\gamma_i$ can be iteratively computed as
$$
\gamma_i' = \begin{cases}
    \hphantom{o}\frac{\epsilon'}{(1 - p) p^k \overline{m}_k} &; i = k\\
    \gamma_{i - 1}' \cdot \frac{i - k + 1}{i p} &; i > k
\end{cases}
$$
and $\gamma_i = \min\{1, \gamma_i'\}$, since the equation above might lead to $\gamma_i' > 1$ and it is important to keep (the upper bound for) $\gamma_i$ as big as possible.

\subsection{Modes of multinomial distribution}
Finding the mode of the multinomial distribution, i.e., the point at which the function $(l_1, \dotsb, l_k) \mapsto \M$ achieves its maximum is beyond trivial~\cite{mode-gall,mode-white}. However, we do not need the exact mode, but only its approximation. Our solution follows~\cite{mode-finucan} and computes the approximation of the inner sum by starting in the point $\lfloor\frac{i}{q_1 + \dotsb + q_k}(q_1, \dotsc, q_k)\rfloor$ where $\lfloor \cdot \rfloor$ is the (component-wise) floor function. Then, we use breadth-first search to find all the points $(l_1, \dots, l_k)$ for which the value of $\M$ is large enough (as explained in the main text, the value decreases when moving away from the mode(s)), until enough elements were calculated for the desired precision of the approximation.

\subsection{Time complexity}
\label{apendix:aprox-complexity}
We will first derive  the time complexity of computing multinomial coefficients. Multinomial coefficients are symmetric in their coefficients. Therefore it is sufficient to calculate $\binom{l_1 + \dotsb + l_k}{l_1, \dotsc, l_k}$ for a partition $(l_1, \dotsc, l_k)$, such that 
$l_1 \geq l_2 \geq \dotsb \geq l_k$.
We will compute such coefficients recursively by using the known formula
$\binom{l_1 + \dotsb + l_k}{l_1, \dotsc, l_k} = \binom{l_1 + \dotsb + l_k}{l_k}\binom{l_1 + \dotsb + l_{k-1}}{l_1, \dotsc, l_{k-1}}$. To do so efficiently, we need to compute all the coefficients
$\binom{l_1 + \dotsb + l_q}{l_1, \dots, l_q}$ for each $q \in \{1, \dots, k-1\}$. On each iteration, we need $\mathcal{O}(l_q)$ steps to compute the binomial coefficient $\binom{l_1 + \dotsb + l_q}{l_q}$ and one additional multiplication to obtain the coefficient $\binom{l_1 + \dotsb + l_q}{l_1, \dots , l_q}$.

To compute all the multinomial coefficients $\binom{l_1 + \dotsb + l_k}{l_1, \dotsc, l_k}$ that are needed in the first $i - k + 1$ iterations of the approximation, we need to compute the coefficients corresponding to ordered partitions $l_1 + \dotsb + l_q = j$, for all $j \in \{1, \dots, i\}$ and for all $q \in \{1, \dots, k\}$. The number of steps needed to compute multinomial coefficients is thus
$$
\bigO\left(
\sum \limits_{j=1}^i \sum \limits_{q=1}^k \sum \limits_{\substack{l_1 + \dotsb + l_q = j \\ l_1 \geq \dotsb \geq l_q}} l_q
\right).
$$
After computing multinomial coefficients, 
we need additional
$$
T(M) = 
\sum \limits_{i=k}^M \sum \limits _{(l_1, \dotsc, l_k) \in S(i)} \left ( \bigO(\log_2 l_1 + \dotsb + \log_2 l_k) + \bigO(k) \right ),
$$
steps to compute the approximation,
since we need $\bigO(\log_2l_j)$ for computing $q_j^{l_j}$ (using fast exponentiation\footnote{For example, to quickly compute $a^{22}$, one first computes the binary representation of $22 = 10110_{(2)}$ and the values $a^2$, $a^4$, $a^8$, and $a^{16}$ (the next term is obtained by squaring the previous one). Then, $a^{22}$ is computed as $a^{22} = a^{16}\cdot a^{4} \cdot a^{2}$.}) and $\bigO(k)$ for the product $\binom{i}{l_1, \dotsc, l_k} q_1 ^{l_1} \dotsm q_k^{l_k}$.
Recall the inequalities $l_j \leq i-k+1$,
and $\gamma_i \leq \epsilon' / ((1 - p)p^i \overline{m}_k {i - 1 \choose k - 1})$, where $\epsilon' = \epsilon / (2(M - k + 1))$. Then,
\begin{align*}
T(M) &\leq \sum \limits_{i=k} ^M k
|S(i)| \left(\bigO(\log_2(i-k+1)) + \bigO(k)\right)\\
&\leq \sum \limits_{i=k}^M \left(1-\frac{\epsilon'}{(1 - p)p^i \overline{m}_k {i - 1 \choose k - 1}}\right) k
\binom{i-1}{k-1}  \left(\bigO(\log_2(i-k+1)) + \bigO(k)\right)   \\
&\leq \left(\bigO(M\log_2 M) + \bigO(M^2)\right) \sum \limits_{i=k}^M \left(1-\frac{\epsilon'}{(1 - p)p^i \overline{m}_k {i - 1 \choose k - 1}}\right) k
\binom{i-1}{k-1} \\
&=  \bigO(M^2)\left(
\sum \limits_{i=k} ^M k
\binom{i-1}{k-1} - \frac{\epsilon'}{(1 - p) \overline{m}_k}\sum_{k = i}^M \frac{1}{p^i}
\right)  \\
&\leq \bigO(M^2)\left( (M - k + 1) \binom{M}{k-1} - \frac{\epsilon'(1 - p^{M - 1 - k})}{(1 - p)^2 p^{M - 2} \overline{m}_k} \right).
\end{align*}

%
%

 \bibliographystyle{elsarticle-num-names} 
 \bibliography{main-cite}






\end{document}